\documentclass[aps,pra,twocolumn,superscriptaddress,10pt,article,nofootinbib,showpacs]{revtex4-1}
\usepackage{amsfonts}
\usepackage{amsmath}
\usepackage{amsthm}
\usepackage{braket}
\usepackage{graphicx}
\usepackage{hyperref}
\usepackage{color}
\usepackage{amssymb}
\usepackage{dsfont}
\usepackage{verbatim}
\usepackage{amsmath,amscd}
\usepackage[all]{xy}

\newcommand{\red}[1]{\textcolor{red}{#1}}
\newcommand{\blue}[1]{\textcolor{blue}{#1}}

\newcommand{\mathscr}[1]{\ensuremath{\mathcal{#1}}}

\newcommand{\zt}{\ensuremath{\mathbb{Z}_2}}
\newcommand{\conj}{\ensuremath{\epsilon}}
\newcommand{\ham}{\ensuremath{\sigma}}
\newcommand{\xfields}{\ensuremath{\phi}}
\newcommand{\zfields}{\ensuremath{\eta}}
\newcommand{\hmodule}{\ensuremath{F}}
\newcommand{\pmodule}{\ensuremath{P}}
\newcommand{\emodule}{\ensuremath{E}}
\newcommand{\proj}{\ensuremath{P}}
\newcommand{\projmod}{\ensuremath{\pi}}
\newcommand{\hilb}{\ensuremath{\mathbb{H}}}

\newcommand{\transp}{\ensuremath{\text{T}}}

\newcommand{\vect}[1]{\ensuremath{{\underline{\mathbf{#1}}}}}
\newcommand{\basis}[1]{\ensuremath{\hat{e}_{#1}}}
\newcommand{\projx}{\ensuremath{P}}
\newcommand{\gauge}{\ensuremath{G}}
\newcommand{\opgauge}{\ensuremath{\mathcal{G}}}
\newcommand{\bdyn}{\ensuremath{\mathcal{B}}}
\newcommand{\opproj}{\ensuremath{\mathcal{P}}}
\newcommand{\supp}{\ensuremath{\text{supp}}}
\newcommand{\kernn}{\ensuremath{\mu}}
\newcommand{\disent}{\ensuremath{U_D}}
\newcommand{\symset}{\ensuremath{M}}
\newcommand{\tr}{\ensuremath{\text{Tr}}}
\newcommand{\sdim}{\ensuremath{d}}
\newcommand{\symtyp}{\ensuremath{T}}

\newcommand{\drawgenerator}[8]{%
\xymatrix@!0{%
& #8 \ar@{-}[ld]\ar@{.}[dd] \ar@{-}[rr] & & #7 \ar@{-}[ld]  \\%
#1 \ar@{-}[rr] \ar@{-}[dd] &  & #2 \ar@{-}[dd] &            \\%
& #6 \ar@{.}[ld] &  & #5 \ar@{-}[uu] \ar@{.}[ll]       \\%
#3 \ar@{-}[rr] &  & #4 \ar@{-}[ru]                       %
}%
}
\newcommand{\drawgeneratorcluster}[9]{%
\xymatrix@!0{%
& #8 \ar@{-}[ld]\ar@{.}[dd] \ar@{-}[rr] & & #7 \ar@{-}[ld]  \\%
#1 \ar@{-}[rr] \ar@{-}[dd] & & #2 \ar@{-}[dd] &            \\%
& #6 \ar@{.}[ld] & #9  & #5 \ar@{-}[uu] \ar@{.}[ll]       \\%
#3 \ar@{-}[rr] &  & #4 \ar@{-}[ru]                       %
}%
}
\newcommand{\drawgeneratorclustertwoa}[9]{%
\xymatrix@!0{%
& #8 \ar@{.}[ld]\ar@{.}[d] \ar@{.}[rrr] & & & #7 \ar@{.}[ld]  \\%
#1 \ar@{.}[rrr] \ar@{.}[ddd] & & \ar@{-}[u] & #2 \ar@{.}[ddd] &            \\%
& \ar@{.}[u] & #9 \ar@{-}[ll] \ar@{-}[r] \ar@{-}[u] \ar@{-}[dd] & \ar@{-}[r] &       \\%
& #6 \ar@{.}[ld] \ar@{.}[r] \ar@{.}[u]  &  \ar@{.}[r] &   & #5 \ar@{.}[uuu] \ar@{.}[l]       \\%
#3 \ar@{.}[rrr] & & & #4 \ar@{.}[ru]                       %
}%
}
\newcommand{\drawgeneratorclustertwob}[9]{%
\xymatrix@!0{%
& #8 \ar@{-}[ld]\ar@{-}[d] \ar@{-}[rrr] & & & #7 \ar@{-}[ld]  \\%
#1 \ar@{-}[rrr] \ar@{-}[ddd] & & \ar@{.}[u] & #2 \ar@{-}[ddd] &            \\%
& \ar@{-}[u] & #9 \ar@{.}[ll] \ar@{.}[r] \ar@{.}[u] \ar@{.}[dd] & \ar@{.}[r] &       \\%
& #6 \ar@{-}[ld] \ar@{-}[r] \ar@{-}[u]  &  \ar@{-}[r] &   & #5 \ar@{-}[uuu] \ar@{-}[l]       \\%
#3 \ar@{-}[rrr] & & & #4 \ar@{-}[ru]                       %
}%
}

\newtheorem{lemma}{Lemma}
\newtheorem{claim}{Proposition}

\begin{document}

\title{Fractal symmetries: Ungauging the cubic code}

\author{Dominic J. \surname{Williamson}}
\affiliation{Vienna Center for Quantum Technology, University of Vienna, Boltzmanngasse
5, 1090 Vienna, Austria}

\begin{abstract}
Gauging is a ubiquitous tool in many-body physics. It allows one to construct highly entangled topological phases of matter from relatively simple phases and to relate certain characteristics of the two. 
Here we develop a gauging procedure for general submanifold symmetries of Pauli Hamiltonians, including symmetries of fractal type. We show a relation between the pre- and post-gauging models and use this to construct short-range entangled phases with fractal-like symmetries, one of which is mapped to the cubic code by the gauging. 
\end{abstract}

\maketitle

\section{ Introduction}
\label{intro}

The study of topological order~\cite{einarsson,Wen90} intertwines many rich areas of physics: strongly correlated quantum many-body condensed-matter systems~\cite{kitaev2006anyons,levin2005string}, quantum codes~\cite{shor1996fault,preskill1998reliable,kitaev2003fault}, topological quantum field theory (TQFT)~\cite{witten1988topological,atiyah1988topological,walker1991witten}, and modular tensor categories~\cite{etingof2005fusion} (and their higher categorical generalizations~\cite{baez1995higher,lurie2009classification,kong2014braided,higherdto}). By viewing the same physics through these complementary lenses valuable insights have been gained.

In two dimensions a concise and deep understanding of topological order has solidified in terms of the theory of anyonic excitations~\cite{wilczek1982magnetic,wilczek1982remarks,moore1989classical,witten1989quantum,turaev1994quantum,kitaev2006anyons,nick,MPOpaper} described mathematically by modular tensor categories~\cite{etingof2005fusion} known to be equivalent to (3-2-1)-extended topological quantum field theories~\cite{bakalov2001lectures}. The theory of stabilizer codes in two dimensions is also well understood~\cite{yoshida2011classification,bombin2014structure}. It is known that all two-dimensional (2D) stabilizer Hamiltonians possess stringlike logical operators and hence are not self-correcting quantum memories at finite temperature~\cite{bravyi2009no}, and furthermore that they are equivalent to some number of copies of the toric code~\cite{bombin2012universal,kubica2015unfolding}. 

In three dimensions the landscape of possibilities remains shrouded in mystery. Progress has been made via the construction of families of fixed point Hamiltonians~\cite{walker20123+,williamson2016hamiltonian} and the development of novel tools such as the 3-loop braiding statistics~\cite{wang2014braiding} (which primarily apply to gauge theories with a possibly anomalous 2D topological boundary) but a general understanding is still lacking. On the other hand novel contributions have been made in assessing the possibility of self-correcting three dimensional (3D) quantum memories~\cite{brell2014proposal,haah2011local,bravyi2013quantum} but a definitive consensus has not been reached. For a nice overview of progress on this topic see the recent review article Ref.\cite{brown2014quantum}. Most significantly this search has revealed models~\cite{haah2011local,yoshida2013exotic} that satisfy conventional definitions of topological order and stability and yet are not described by any (conventional) TQFT. Rather than being fixed points under real space blocking renormalization group flow they may bifurcate into multiple copies of themselves~\cite{haah2014bifurcation}.

The idea of gauging pervades the literature on topological order in condensed matter systems~\cite{wilson1974confinement,kogut1975hamiltonian,levin2012braiding,Gaugingpaper}. This process makes global symmetries local while allowing one to relate certain physical properties of the pre- and post-gauged systems~\cite{levin2012braiding,Gaugingpaper,williamson2014matrix}. This is most commonly applied to a truly global and on site symmetry, although it has also been adapted to higher form symmetries~\cite{gaiotto2015generalized,yoshida2015topological} that are important in the classification of higher dimensional phases of matter~\cite{kong2014braided,thorngren2015higher}. 

In this paper we develop a framework for gauging submanifold symmetries, including those of fractal type, using the language of translationally invariant stabilizer Hamiltonians~\cite{gottesman1997stabilizer,haah2013commuting,haah2013lattice}. We then demonstrate relations between physical characteristics of the pre- and post- gauged models. 
Our formalism includes exotic examples such as Haah's cubic code~\cite{haah2011local} and more conventional examples such as generalized toric codes. 
We go on to use the tools developed to construct novel cluster states with fractal-type symmetries.

The results presented here share many similarities with those in the concurrent work Ref.\cite{Haahnew}. 
In Ref.\cite{Haahnew} the authors also develop a gauging duality map and apply it to study what they call fracton topological orders, meaning those with pointlike excitations that are not created by stringlike operators. 
The fundamental idea underlying the construction of their gauging map, which they call ``\emph{F-S}'' duality, is the introduction of gauge degrees of freedom to mediate many-body interactions.  This is identical to the ideology of our approach, although the execution and applications differ. 
They explicitly describe their gauging duality for many examples of classical spin models with what they call subsystem symmetries (which we refer to as submanifold symmetries) including Haah's cubic code~\cite{haah2011local} and the model due to Chamon, Bravyi, Leemhuis and Terhal~\cite{chamon2005quantum,bravyi2011topological}. 
While the set of examples they consider differs from those here, they are all equally well described using our formalism. 
Moreover the gauging duality map of Ref.\cite{Haahnew} is defined only for classical Ising models and does not allow any local symmetries, which excludes them from gauging the 1-form symmetry of the toric code for example. An advantage of our approach is that it allows one to gauge any quantum model that is described by a spatially local Hamiltonian with a given submanifold symmetry; this goes beyond the classical (diagonal in the computational basis) Ising models considered in Ref.\cite{Haahnew}.

\section{Background}

 In this section we recount several basic notions from the polynomial formalism developed in Refs.\cite{haah2013lattice,haah2013commuting}.

Our focus is on Pauli Hamiltonians that are local, translation invariant and consist of a sum of terms that are each tensor products of exclusively Pauli $X$ or $Z$ matrices (We shall loosen the last requirement somewhat in the next section). We use the language of polynomials developed by Haah in Refs.\cite{haah2013lattice,haah2013commuting} as it provides a succinct description of the operators in this setting. 
In Haah's formalism a Pauli operator is specified by a column of polynomials over $\zt$. For a translationally invariant system with $\sdim$ spatial dimensions each lattice site is specified by a vector $\vect{i}\in\mathbb{Z}^\sdim$; when there are $Q$ qubits per site a single qubit is specified by a pair $(\vect{i},q)$ for $q\in\{1,\dots,Q\}$. A general Pauli operator is then mapped to a column of length $2Q$ with a multivariate polynomial over $\mathbb{Z}_2$ in each entry as follows:
\begin{align}
\bigotimes_{\vect{i},q} X^{p^q_{\vect{i}}}_{\vect{i},q} \bigotimes_{\vect{i},q} Z^{r^q_{\vect{i}}}_{\vect{i},q} \mapsto 
\begin{pmatrix}
 \mathbf{p}
\\ \hline
\mathbf{r}
 \end{pmatrix}
\end{align}
where $\mathbf{p}=(p^1,\dots,p^Q)$ is a column consisting of entries $p^q$ which are multivariate polynomials over $\zt$ whose $\vect{x}^\vect{i}$ coefficient is given by $p^q_\vect{i}\in\{0,1\}$ (we are using multi-index notation) i.e.
\begin{align}
p^q=\sum_{\vect{i}\in\mathbb{Z}^d}  p^q_\vect{i}\ x_1^{i_1} \cdots x_Q^{i_Q}
\end{align}
with similar notation for $\mathbf{r}$.
For example on a two dimensional lattice with two types of qubits $r,b$ per site the operator $X_{(0,0),r}X_{(0,1),r}X_{(1,1),b}Z_{(1,0),r}$  is specified by the polynomial $(1+y,  xy,  \,\vline  x, 0)$ as shown in by
\begin{align}
 \xymatrix@!0{
\red{X}\blue{I}  \ar@{-}[r]  & \red{I}\blue{X}\ar@{-}[d]\\
\red{X}\blue{I} \ar@{-}[u]   & \red{Z}\blue{I}\ar@{-}[l] 
 } 
 \quad \quad
  \quad \quad
  \xymatrix@!0{
 y \ar@{-}[r]  & xy \ar@{-}[d]\\
 1 \ar@{-}[u]   & x \ar@{-}[l] 
 }
\end{align}
A pair of Pauli operators 
$\begin{pmatrix}
 \mathbf{p}
\\ \hline
\mathbf{r}
 \end{pmatrix},\begin{pmatrix}
 \mathbf{s}
\\ \hline
\mathbf{v}
 \end{pmatrix}$ commutes iff their symplectic inner product is zero, i.e.
\begin{align}
\left[
\begin{pmatrix}
\bar{\mathbf{s}}^\transp  &\vline \  \bar{\mathbf{v}}^\transp 
\end{pmatrix}
\lambda_Q
\begin{pmatrix}
 \mathbf{p}
\\ \hline
\mathbf{r}
 \end{pmatrix}
 \right]_{\vect{0}} =0
 \end{align}
where $\lambda_Q:=(ZX)\otimes \openone_Q $ is the relevant symplectic form, $\bar{s}$ is the antipode map (sending each monomial summand to its inverse) and the subscript $[\cdot]_\vect{0}$ denotes the constant term of the polynomial. 
For convenience we also define the conjugation operation $\begin{pmatrix}
 \mathbf{s}
\\ \hline
\mathbf{v}
 \end{pmatrix}^\dagger  := {(\bar{\mathbf{s}}^\transp \, \vline \, \bar{\mathbf{v}}^\transp)}$. 

In Haah's formalism the Hamiltonian is identified with a module generated by the stabilizers on the unit cell. More specifically let $\hmodule$ be a free module of rank $T$, think of this as the set of position labels for individual stabilizer terms, and $\pmodule$ be the module of Pauli matrices on the lattice. The Hamiltonian module with $T$ types of local interaction terms $\left\{\begin{pmatrix}
 \mathbf{p}_t
\\ \hline
\mathbf{r}_t
 \end{pmatrix}\right\}_t$ is generated by 
$$\ham:=\begin{pmatrix}
\mathbf{p}_1 & \dots & \mathbf{p}_T
\\
\hline
\mathbf{r}_1 & \dots & \mathbf{r}_T
\end{pmatrix}
$$
 which maps $\ham : \hmodule \rightarrow \pmodule$. 
 Its symplectic conjugate 
$\conj:=\ham^\dagger \lambda_Q$ maps $\conj:\pmodule \rightarrow \emodule$, that is from the Pauli module $\pmodule$ to the virtual excitation module $\emodule$; think of this as the positions of various stabilizer terms that anticommute with a given Pauli operator. 
The condition that the Hamiltonian is commuting and hence defines a stabilizer code is simply $\conj \ham = 0 $ which is equivalent to the sequence 
$$\begin{CD} \hmodule @>\ham >> \pmodule @>\conj>> \emodule\end{CD}$$ 
forming a complex. 
It was shown in Refs.\cite{haah2013commuting,haah2013lattice} that the stabilizer Hamiltonian is topologically ordered if the aforementioned sequence is exact, i.e. $\text{im}(\sigma)=\text{ker}(\epsilon)$. 

In the case of Calderbank-Shor-Steane (CSS) codes this complex breaks up into a direct sum since we have $\ham=\ham_X\oplus\ham_Z$ and the commutation condition becomes ${\ham}_Z^\dagger\ham_X=0={\ham}_X^\dagger \ham_Z$.

In terms of the bipartite interaction graph of the Hamiltonian, $\ham$ can be thought of as mapping from a Hamiltonian node to the qubit nodes in its support (as an operator) and $\epsilon$ can be though of as mapping from a qubit node to the adjacent Hamiltonian nodes with which a Pauli operator on that qudit anticommutes. Note we in fact need to add extra structure to distinguish the $X$ and $Z$ terms in the interaction graph above, alternatively if the Hamiltonian is CSS we can consider separate $X$ and $Z$ interaction graphs, corresponding to $\ham_X$ and $\ham_Z$, and the only relevant operators are then either $Z$ or $X$ respectively.

\section{Gauging}

In this section we build up a procedure for gauging submanifold symmetries and analyze the important properties of this gauging map. 
We start by specifying the type of translationally invariant, symmetric, local Hamiltonians we treat. We then move on to the definition of the gauging procedure and proofs of several results that demonstrate its key features. Finally we describe the relationship between the gauging procedure and translationally invariant, local, CSS stabilizer Hamiltonians and give a construction of cluster state~\cite{briegel2001persistent} models with submanifold symmetries.

\subsection{Hamiltonian construction}

Consider a system of ``matter'' degrees of freedom ($Q$ qubits per site) with Hilbert space $\hilb_\text{m}$ governed by a translationally invariant local Hamiltonian 
$$H_\text{m}=\sum\limits_{\vect{i}\in \mathbb{Z}^d}\sum\limits_{k} h_{\vect{i},k}$$
 with a family of on-site symmetry operators generated by a tensor product of $X$ on qubits contained in each closed submanifold of some fixed but arbitrary dimension, possibly fractals with noninteger dimension (note these manifolds only appear as discretizations with a minimum length scale cutoff). 
 
We only consider hypercube local Hamiltonian terms, hence the $X$ symmetry can be specified by hypercube local constraints expressed as products of $Z$ fields that commute with any symmetry operator. 
These $Z$ constraints can be understood as locally checking whether an operator, which is a tensor product of $X$s, has the shape of an appropriate submanifold on which it is a symmetry.
Associated with these checks is a fundamental object in our framework, the map 
$$\zfields:\hmodule_\symtyp\rightarrow P$$
 from $\hmodule_\symtyp$, a free module of rank $\symtyp$, to $P$, the Pauli module ($\symtyp$ is the number of independent local $Z$ constraints).

In addition to sums and products of these $Z$ constraint fields, a symmetric Hamiltonian may contain arbitrary $X$ perturbations. There are two important irreducible types of $X$ fields, single site $X$ fields and hypercube local $X$ fields that commute with the $Z$ constraints. 
Hence the set of symmetric field perturbations we consider break up into the local $Z$ constraint fields described by $\zfields$, the single $X$ terms described by $\openone_Q$ and possibly a number $S_X$ of additional $X$ fields described by a map 
$$\xfields: \hmodule_{S_X} \rightarrow P$$
 from a free module of rank $S_X$ to the Pauli module, which satisfies $\xfields^\dagger \lambda_Q \zfields =0$.
These local $X$ fields that commute with the $Z$ constraints are in fact local symmetries of the model, we will largely ignore them for the time being as they become trivial after gauging.

In summary we are considering Hamiltonians that commute with a set of tensor product $X$ operators which might best be described as a \emph{locally defined} symmetry. Such symmetries are concretely defined in terms of a chain complex
\begin{equation}
\begin{CD} \hmodule_{S_X} @>\xfields >> \pmodule @>\zfields^\dagger \lambda_Q>> \hmodule_T \end{CD}
\end{equation}
with local symmetries given by the image of $\xfields$ and equivalence classes of global symmetries given by the distinct homology classes of the sequence. Note this homology description of a locally defined, tensor product $X$ symmetry is very general and does not rely on a translationally invariant structure or a fixed spatial dimension. This may prove interesting for future work.

An illustrative example is a generalized toric code in $d$ spatial dimensions with qubits on $k$-cells, $X$ stabilizers on $(k-1)$-cells and $Z$ stabilizers on $(k+1)$-cells. This model has $(d-k)$-manifold ($k$-form) $X$ symmetry, specified by the local $Z$ stabilizer constraints on the unit cell and their translations (note this analysis extends to an arbitrary cellulation of a closed $d$-manifold). See Sec.~\ref{toric code} for a more detailed description of the 2D case.

\subsection{Gauging procedure}\label{gauging procedure}

In this section we follow and generalize the approach of Ref.\cite{Gaugingpaper} to produce a (nearly) unambiguous gauging procedure for quantum states and operators with submanifold symmetry.

To gauge the Hamiltonian $H_\text{m}$ we must first specify  the gauge degrees of freedom. We extend the canonical choice for gauging a conventional $k$-form symmetry, given by associating a gauge field to each $(k+1)$-cell, with a recipe that also deals with more exotic cases. The gauge and matter Hilbert space $\hilb_\text{m}\otimes\hilb_\text{g}$ is built by tensoring in a gauge qubit for each $Z$ constraint field, i.e. each label in $\hmodule_\symtyp$. 
The locality of this system is described by the bipartite interaction graph of the $Z$ constraint fields, which is generated by $\zfields$. 
Hence $\zfields$ can now be thought of as the map from gauge qubits to neighbouring matter qubits, $\zfields: \hmodule_{T} \rightarrow \hmodule_{Q}$, we will continue to use this definition below.

The next ingredient in the gauging procedure is a set of local constraints that project onto states satisfying a $\mathbb{Z}_2$ ``Gauss law''. This law states that the charge on each matter qubit equals the sum of the fields on the neighbouring gauge qubits.
The local gauge constraints are generated by the map
\begin{align}
\projmod :=
\begin{pmatrix}
\mathbf{\openone_Q} \\
\zfields^\dagger \\
\hline
0
\\
0
\end{pmatrix}.
\end{align} 
Specifically each constraint  is given by a projector 
$$\proj_{\vect{i},q}:=\frac{1}{2}(\projx_{\vect{i},q}(0)+\projx_{\vect{i},q}(1) \,)$$
 onto the +1 eigenspace of the Pauli operator $\projmod\, \vect{x}^\vect{i}\, \basis{q}$,
 which is identified with $\projx_{\vect{i},q}(1)$,
  where $\basis{q}$ is the column with a $1$ in the $q$th entry and zeros elsewhere and $\projx_{\vect{i},q}(0):=\openone$. 
  The full projector onto the gauge invariant subspace is then given by the product of these local projectors 
  $$\proj=\prod\limits_{\vect{i},q} \proj_{\vect{i},q}.$$
The state gauging map $\gauge:\hilb_\text{m} \rightarrow \hilb_\text{m}\otimes\hilb_\text{g}$ is given by 
$$\gauge \ket{\psi} := \proj \ket{\psi} \otimes \ket{0}^{\otimes N T}$$
 where $N$ is the number of unit cells in the system. 
 
 The local projection of an operator onto the gauge invariant subspace is given by 
$$\opproj_\Gamma[\cdot]:= \sum\limits_{\mathcal{S}_\Gamma} \bigotimes\limits_{\vect{i},q} \projx_{\vect{i},q}(s_\vect{i}^q)\big{|}_{\Gamma} \ [\cdot] \  \bigotimes\limits_{\vect{i},q} \projx_{\vect{i},q}(s_\vect{i}^q)\big{|}_{\Gamma}$$
 where $\big{|}_{\Gamma}$ denotes the restriction of an operator onto the qubits within the region $\Gamma$ and 
 the sum is over the set of variables ${\mathcal{S}_\Gamma:=\{s_\vect{i}^q\}\in\{0,1\}^{|\Gamma|}}$.
Then the corresponding operator gauging map is defined by 
$$\opgauge[O]:=\opproj_\Gamma[O\bigotimes\limits_{v\in \Gamma}\ket{0}\bra{0}_v]$$ 
where $\Gamma$ is a (minimal) region containing $\supp(O)$ that is generated by a set of points corresponding to gauge qubits in the interaction graph and their neighbours under $\zfields$. 

The matter Hamiltonian is gauged in a locality preserving way as follows: 
$$H_\text{m}^{\opgauge}:=\sum \limits_{\vect{i},k} \opgauge[h_{\vect{i},k}]\, .$$ 
To construct the full gauged Hamiltonian in a nontrivial way we must also specify some fields $H_\bdyn$ that describe dynamics of the gauge spins. 
We will introduce these fields in the zero gauge coupling limit where the gauge degrees of freedom are frozen to have `zero magnetic flux' (this is analogous to a flatness condition on a finite group connection) perturbations away from this point are considered later. Note the $H_\bdyn$ fields are only defined within the gauge invariant subspace and so should commute with all local gauge constrains. 
The $Z$ fields commuting with all $\proj_{\vect{i},q}$ are precisely those described by polynomials in the kernel of $\zfields$, in addition we require them to be a set of independent generators that are local to a hypercube and are hence described by a map $\kernn$ which generates the kernel of $\zfields$.  

The full gauged Hamiltonian is then given by
$$H_\text{full}:=H_\text{m}^\opgauge+\Delta_\bdyn H_\bdyn + \Delta_{\proj} H_\proj$$ 
where $\Delta_\bdyn,\Delta_P>0$
 and $H_P$ is the sum of all local gauge projectors. From the definitions of the various Hamiltonian terms one can see that 
 $$[H_\text{m},H_\bdyn]=[H_\bdyn,H_P]=[H_P,H_\text{m}]=0.$$ 
 For $\Delta_P$ sufficiently large the low energy subspace of this Hamiltonian is gauge invariant, with a true gauge theory being recovered in the limit $\Delta_P\rightarrow \infty$. When $\Delta_\bdyn$ is also sufficiently large the states relevant to the low energy physics are those within the gauge invariant subspace that also have `flat' gauge connections (specified by the $\kernn$ constraints).

This full gauged Hamiltonian is equivalent, via a constant depth circuit $\disent$ of local isometries, to another Hamiltonian where the gauge has been fixed to remove the local gauge constraints thus restoring a clear tensor product structure to the gauge invariant physics. 
The circuit is constructed from a product of controlled-$X$ gates from each matter qubit to each of its adjacent gauge qubits (under the map 
$\zfields^\dagger$). Note this unitary disentangles each local gauge constraint $\proj_{\vect{i},q}$ such that it becomes a projector onto the $\ket{+}$-state of the single qubit at site $(\vect{i},q)$. Hence the full disentangling isometry $\disent$ is given by 
\begin{align}
\disent:=\bigotimes_{\vect{i},q}\bra{+} \, \prod_{\vect{i},q} \, \prod_{\vect{x}^\vect{j} \basis{t} \in \zfields^\dagger \vect{x}^\vect{i}\basis{q}} CX_{(\vect{i},q)\rightarrow (\vect{j},t)} \, .
\end{align}
Now the disentangled Hamiltonian, which acts purely on the gauge qubits remaining, is given by 
$$\disent H_\text{full} \disent^\dagger=\hat{H}_\text{m}^\opgauge+\Delta_\bdyn H_\bdyn$$ 
where $\hat{H}_\text{m}^\opgauge=\sum \limits_{\vect{i},k} \disent \opgauge_{\Gamma^k_{\vect{i}}}[h_{k,\vect{i}}] \disent^\dagger$ is again a sum of local terms. 

We close the section by giving a summary of the full gauging and disentangling procedure in terms of its effect on local symmetric Pauli terms. In the polynomial language these are as follows
\begin{align}
\vect{x}^\vect{i} \basis{q} &\mapsto \zfields^\dagger \lambda_Q \vect{x}^\vect{i} \basis{q} \label{xeq1}
\\
\zfields \vect{x}^\vect{i} \basis{t} 
&\mapsto 
\vect{x}^\vect{i} \basis{t} \label{zeq1} .
\end{align}
Eq.\eqref{xeq1} describes the mapping of a single qubit $X$ to a product of $X$ terms on the neighbouring gauge qubits and Eq.\eqref{zeq1} describes the mapping of a minimal symmetric $Z$ field (which is necessarily generated by $\zfields$) to a single $Z$ on the corresponding gauge qubit. Note these mappings suffice to describe the transformation of all local symmetric tensor products of Pauli matrices. 
In general the choice of $\mathbf{p}$ generating a set of symmetric $Z$ fields $\zfields \mathbf{p}$ may not be unique since the kernel of $\zfields$ may be nontrivial. The exact term obtained is determined by the local support set $\Gamma$ that is chosen when gauging the symmetric $Z$ fields. All choices for gauging this term are related by some local fields in $H_\bdyn$ since it is generated by $\kernn$ which also generates the kernel of $\zfields$. Hence all such choices have an equivalent action upon the ground space provided $\Delta_\bdyn$ is sufficiently large.

In summary we have constructed a bipartite graph determined by the $X$ symmetry which is specified in the polynomial language by $\zfields$. The gauging procedure sends single site $X$ terms to a product of $X$s on the neighbouring gauge bits and symmetric $Z$ terms (which necessarily lie in the image of $\zfields$) to a local term in their preimage under $\zfields$. Additional local $Z$ fields were also introduced in terms of the map $\kernn$ which generates the kernel of $\zfields$. Hence we have a CSS stabilizer Hamiltonian specified by the \emph{gauging complex }
\begin{align}
\begin{CD}
\hmodule @>\ham>> \pmodule @>\conj>> \emodule
\end{CD}
\end{align}
where $\ham=\zfields^\dagger \oplus \kernn$ and $\conj=\ham^\dagger\lambda_Q$. Note the additional $Z$ terms are noncommuting perturbations to this code Hamiltonian.

\subsection{Basic properties of the gauging procedure}

The mantra of gauging is `global symmetry to local symmetry'. This is made precise in the gauging procedure above as follows; any symmetry specified by a subset of qubits $\symset$ of the original model 
$$X(\symset):=\bigotimes\limits_{(\vect{i},q)\in\symset} X_{\vect{i},q}$$
 can be reconstructed from the local symmetries of the gauged model, i.e.
 $$\prod\limits_{(\vect{i},q)\in\symset} \projx_{\vect{i},q}(1) = X(\symset)\otimes \openone_\text{g}$$
 where $\text{g}$ indicates the gauge subsystem.

Gauging in the zero coupling limit (described above in terms of the maps $\gauge,\,\opgauge$) provides an equivalence between the gauged and ungauged models in that the operator gauging map is invertible (in a sense) and furthermore all symmetric expectation values are preserved. We proceed to show this below. 
For the remainder of the section we use the labeling convention that $(\vect{j},t)$ are gauge qubits while $(\vect{i},q)$ are matter qubits. 

\begin{claim} \label{lem1}
The operator gauging map is invertible for symmetric operators $O$ in the following sense 
$\emph{\tr}_{(\vect{j},t)\in\Gamma} \left(\opgauge[O] \cdot \bigotimes_{(\vect{j},t)\in\Gamma} \ket{0}\bra{0} \right)=O$.
\end{claim} 
\begin{proof}
This is simply a calculation
\begin{align*}
&{\tr}_{(\vect{j},t)\in\Gamma} \left(\opgauge[O] \cdot \bigotimes_{(\vect{j},t)\in\Gamma} \ket{0}\bra{0} \right) 
\\
&= 
\sum_{\mathcal{S}_\Gamma}\bigotimes_{(\vect{i},q)\in\Gamma} X_{\vect{i},q}^{s_i^q} \ O \bigotimes_{(\vect{i},q)\in\Gamma} X_{\vect{i},q}^{s_i^q} \ \prod_{(\vect{j},t)\in\Gamma} \delta \left(\sum_{\zfields^\dagger \vect{x}^{\vect{i}} \basis{q} \ni \vect{x}^\vect{j} \basis{t}}  s_\vect{i}^q \,\right)
\\
&=O 
\end{align*}
the only nontrivial step is realizing that whenever the variables $s_\vect{i}^q$ satisfy the $\delta$ condition the operator $\bigotimes\limits_{(\vect{i},q)\in\Gamma} X_{\vect{i},q}^{s_i^q}$ is a symmetry and hence commutes with $O$ by assumption.
\end{proof}

\begin{lemma}  \label{lem2}
The operator $\gauge^\dagger \gauge$ projects onto the symmetric subspace.
\end{lemma}
\begin{proof}
Again this is simply a calculation
\begin{align*}
\gauge^\dagger \gauge = 
\sum_{\mathcal{S}_\Lambda} \sum_{\bar{\mathcal{S}}_\Lambda} \bigotimes_{(\vect{i},q)\in\Lambda} X_{\vect{i},q}^{s_\vect{i}^q+\bar{s}_\vect{i}^q} \ \prod_{(\vect{j},t)} \delta \left( \sum_{\zfields^\dagger \vect{x}^{\vect{i}} \basis{q} \ni \vect{x}^\vect{j} \basis{t}}  s_\vect{i}^q+\bar{s}_\vect{i}^q \,\right)
\end{align*}
where $\Lambda$ is the full interaction graph (note the sums are only over matter qubits $(\vect{i},q)$ and the product is over the gauge qubits $(\vect{j},t)$). Observe that due to the $\delta$ condition this is by definition the sum over all symmetries of the model. 
\end{proof}

\begin{lemma}  \label{lem3}
The identity $\opgauge[O]\gauge=\gauge O$ holds for any symmetric local operator $O$.
\end{lemma}
\begin{proof}
Once again this is simply a calculation 
\begin{align*}
&\opgauge[O]\gauge =
\\
&
\sum_{\bar{\mathcal{S}}_\Gamma}   \bigotimes_{(\vect{i},q)\in\Gamma} X_{\vect{i},q}^{\bar{s}_\vect{i}^q} 
\ O 
 \bigotimes_{(\vect{i},q)\in\Gamma} X_{\vect{i},q}^{\bar{s}_\vect{i}^q} 
 \bigotimes_{(\vect{j},t)\in\Gamma}  \left|\sum_{\zfields^\dagger \vect{x}^{\vect{i}} \basis{q} \ni \vect{x}^\vect{j} \basis{t}}  \bar{s}_\vect{i}^q \,\right\rangle
 \\
 &
\left\langle\sum_{\zfields^\dagger \vect{x}^{\vect{i}} \basis{q} \ni \vect{x}^\vect{j} \basis{t}}  \bar{s}_\vect{i}^q \,\right|
\sum_{{\mathcal{S}}_\Lambda} \bigotimes_{(\vect{i},q)\in\Lambda}  X_{\vect{i},q}^{{s}_\vect{i}^q} 
  \bigotimes_{(\vect{j},t)\in\Lambda}  \left|{\sum_{\zfields^\dagger \vect{x}^{\vect{i}} \basis{q} \ni \vect{x}^\vect{j} \basis{t}}  s_\vect{i}^q \,}\right\rangle
\\ 
&
=
\sum_{\bar{\mathcal{S}}_\Gamma} \sum_{{\mathcal{S}}_\Lambda} \bigotimes_{(\vect{i},q)\in\Lambda} 
 X_{\vect{i},q}^{{s}_\vect{i}^q} \  \bigotimes_{(\vect{i},q)\in\Gamma} X_{\vect{i},q}^{s_\vect{i}^q+\bar{s}_\vect{i}^q} 
\ O 
 \bigotimes_{(\vect{i},q)\in\Gamma} X_{\vect{i},q}^{s_\vect{i}^q+\bar{s}_\vect{i}^q} 
 \\
 &
\prod_{(\vect{j},t)\in\Gamma}  \delta \left(\sum_{\zfields^\dagger \vect{x}^{\vect{i}} \basis{q} \ni \vect{x}^\vect{j} \basis{t}}  s_\vect{i}^q+\bar{s}_\vect{i}^q \, \right) 
\  \bigotimes_{(\vect{j},t)\in\Lambda}  \left|{\sum_{\zfields^\dagger \vect{x}^{\vect{i}} \basis{q} \ni \vect{x}^\vect{j} \basis{t}}  s_\vect{i}^q \,}\right\rangle
\\
&
=G O
\end{align*}
where the final step follows from the $\delta$ condition and the symmetry of $O$.
\end{proof}

\begin{claim}
Any matrix element of a local symmetric operator $O$ taken with respect to a symmetric state $\ket{\psi_0}$ and an arbitrary state $\ket{\psi_1}$ is preserved by the gauging procedure i.e. $\bra{\psi_0} O \ket{\psi_1}=\bra{\psi_0}\gauge^\dagger\, \opgauge[O] \, \gauge\ket{\psi_1}$.
\end{claim}
\begin{proof}
We have
\begin{align*}
\bra{\psi_0}\gauge^\dagger\, \opgauge[O] \, \gauge\ket{\psi_1} &=\bra{\psi_0}\gauge^\dagger\gauge\, O\ket{\psi_1} \\
&=\bra{\psi_0} O \ket{\psi_1}
\end{align*}
where the first equality follows from Lemma~\ref{lem3} and the second from Lemma~\ref{lem2}.
\end{proof}

\begin{lemma} \label{lem4}
The states $\{\gauge \ket{\lambda}\}$, for a basis $\{\ket{\lambda}\}$ of $H_\text{m}$, span the ground space of $\Delta_\bdyn H_\bdyn+\Delta_P H_P$ with $\Delta_\bdyn,\Delta_P>0$. 
\end{lemma}
\begin{proof}
As discussed in Section~\ref{gauging procedure} the ground space of $H_P$ is spanned by the states $\{P \ket{\lambda}\otimes\ket{\psi}\}$ for bases $\{\ket{\lambda}\},\{\ket{\psi}\}$ of $\hilb_\text{m},\hilb_{\text{g}}$  respectively. 
To restrict to the ground space of $H_\bdyn$ we consider the computational basis for $\hilb_\text{g}$ which consists of states $\ket{\mathcal{S}_\Lambda}=\bigotimes\limits_{\vect{i},q}X^{s^q_\vect{i}}_{\vect{i},q}\, \ket{0}^{\otimes |\Lambda|}$. Since each local field in $H_\bdyn$ commutes with $P$ the combined ground space is spanned by states $\{ P \ket{\lambda} \otimes \ket{\mathcal{S}_\Lambda}\} $ where $[\bigotimes\limits_{\vect{i},q}X^{s^q_\vect{i}}_{\vect{i},q},H_B]=0$; i.e. the $X$ terms correspond to polynomials in the kernel of $\mu^\dagger$. 
We only treat the exact case where $\text{ker}(\mu^\dagger)=\text{im}(\zfields^\dagger)$ (this is always true for our constructions from topologically ordered CSS codes see Sec.~\ref{CSS}, if this assumption is loosened one must deal more carefully with the ground space~\cite{williamson2014matrix}) then the only relevant states in $\hilb_\text{g}$ are generated by a Pauli operator of the form $\zfields^\dagger \mathbf{p}$. 

Note we have the relation $P P_{\vect{i},q}(1) = P$ hence 
\begin{align} \label{Prel}
P\, P_{\vect{i},q}(1)\big{|}_\text{g} = P\, P_{\vect{i},q}(1)\big{|}_\text{m}  
\end{align}
where $\big{|}_\text{g/m}$ denotes the restriction of the operator onto the gauge or matter qubits respectively. 
Since any Pauli operator specified by $\zfields^\dagger \mathbf{p}$ is of the form $\prod\limits_{\vect{i},q}  P_{\vect{i},q}(1)\big{|}_\text{g}$ the ground space of $\Delta_\bdyn H_\bdyn+\Delta_P H_P$ is spanned by states 
\begin{align*}
P \ket{\lambda}\otimes  \prod\limits_{\vect{i},q}  P_{\vect{i},q}(1)\big{|}_\text{g} \ket{0}^{\otimes |\Lambda|} &
\\
=
 P  \prod\limits_{\vect{i},q}&  P_{\vect{i},q}(1)\big{|}_\text{m} \ket{\lambda}\otimes \ket{0}^{\otimes |\Lambda|}
 \\
 = \gauge \prod\limits_{\vect{i},q}&  P_{\vect{i},q}(1)\big{|}_\text{m} \ket{\lambda}
\end{align*}
where we have used Eq.\eqref{Prel}.
Hence the ground space is spanned by states of the form $\{ \gauge \ket{\lambda} \}$. 
\end{proof}

\begin{claim} \label{ggap}
The gauging procedure preserves a gap; i.e. if $H_\text{m}$ has a uniform constant energy gap then $H_\text{full}$ does too, provided the constants $\Delta_\bdyn,\Delta_P$ are sufficiently large.
\end{claim}
\begin{proof}
By Lemma~\ref{lem4}, for $\Delta_\bdyn,\Delta_P>0$ sufficiently large, the ground space of $H_\text{full}$ is spanned by states of the form $\{ \gauge \ket{\lambda} \}$. Since $H_\text{m}^\opgauge=\sum \limits_{\vect{i},k} \opgauge_{\Gamma^k_{\vect{i}}}[h_{\vect{i},k}]$ is a sum of gauged local operators Lemma~\ref{lem3} implies that, for any matter eigenstate $H_\text{m}\ket{\lambda}=\lambda \ket{\lambda}$, we have $H_\text{m}^\opgauge \gauge\ket{\lambda}=\lambda \gauge\ket{\lambda}$. 
Hence $H_\text{full}$ has the same lowest eigenvalue as $H_\text{m}$ (assuming a symmetric ground state) and gap $\Delta_\text{full}\geq\min(\Delta_\text{m},\Delta_\bdyn,\Delta_P)$.
\end{proof}

We remark that Proposition~\ref{ggap} implies that gauging defines a function from the set of gapped phases of the ungauged model into the set of gapped phases of the gauged model. That is, Hamiltonians from the same symmetry protected phase must land in the same phase of the gauged model. 

\subsection{Properties of the gauging complex}

Recall the \emph{gauging complex} 
\begin{align}
\begin{CD}
\hmodule @>\ham>> \pmodule @>\conj>> \emodule
\end{CD}
\end{align}
defined in terms of the maps involved in the gauging procedure $\ham=\zfields^\dagger\oplus \kernn$.
By focusing on the maps in the gauging complex one can infer interesting relationships between quantities pre- and post- gauging. 

Firstly any product of the generating symmetric $Z$ fields which multiplies to identity in the initial model gives an element $\mathbf{p}\in\text{ker}(\zfields)$ and hence a $Z$ symmetry of the gauged model. Furthermore when the gauging complex is exact and the gauged model is topologically ordered $\text{ker}(\zfields)$ is generated by the map $\kernn$ and hence $\mathbf{p}=\kernn\, \mathbf{r},\, \exists \mathbf{r}$. Then $\kernn$ describes the minimal local $Z$-fields that commute with the gauged $X$ terms.

Secondly any $X$ symmetry of the initial model is an element $\mathbf{p}\in\text{ker}(\zfields^\dagger)$ which specifies a product of $X$ stabilizers equal to the identity, i.e. a redundant $X$ stabilizer, in the gauged model. This is relevant in the calculation of the number of qubits encoded into the ground space of the gauged model, which also requires information about redundant $Z$ stabilizers. 

Notice that the gauging procedure is in fact a duality map, in that applying it twice takes us back to the original model. 
To achieve this duality we consider gauging the $Z$ symmetry,  generated by $\kernn$, of the gauged model. 
The local $X$ fields commuting with this symmetry are generated by $\zfields^\dagger$, any product of them equal to identity is in $\text{ker}(\zfields^\dagger)$ by definition and gives a symmetry of the twice gauged model. Let $\xfields$ be a local map generating $\text{ker}(\zfields^\dagger)$, then $\xfields$ describes the independent local $X$ fields that commute with the $Z$ stabilizers of the twice gauged model. These twice gauged $Z$ terms are given by $\zfields$. Note these are precisely the local commuting $Z$ and $X$ fields in the the initial stabilizer Hamiltonian. 
This suggests an addition to the picture of the gauging complex, completing the circle of gauging
\begin{align}
\begin{CD}
\hmodule @>\ham>> \pmodule @>\conj>> \emodule \\
@| @. @|\\
\hat\emodule @<<\hat{\conj}< \hat\pmodule @<<\hat{\sigma}< \hat\hmodule 
\end{CD}
\end{align}
where $\hat{\sigma}:=\xfields\oplus\zfields, \, \&\ \hat{\conj}=\hat{\sigma}^\dagger\lambda_Q$.

Collecting these facts together, we note the number of encoded qubits in the ungauged model is 
$$N[Q-T+S_Z-S_X]+C_\text{m}$$ 
where $N$ is the number of unit cells, $Q$ is the number of matter qubits per site, $T=\text{rank}(\zfields)$ is the number of local $Z$ stabilizers, $S_Z=\text{rank}(\kernn)$ is the number of redundant $Z$ stabilizers locally, $S_X=\text{rank}(\xfields)$ is the number of independent local $X$ symmetries and $C_\text{m}$ accounts for global products of $X$  and $Z$ stabilizers that multiply to the identity upon taking closed bounday conditions for the matter model.
The number of encoded qubits in the gauged model is given by 
$$N[T-Q+S_X-S_Z]+C_\text{g}$$ 
where $T$ now corresponds to the number of gauge qubits per site, $Q$ is the number of local $X$ stabilizers, $S_X$ is the number of redundant $X$ stabilizers locally, $S_Z$ is the number of independent local $Z$ stabilizers and again $C_\text{g}$ accounts for global products of $X$  and $Z$ stabilizers that multiply to the identity upon taking closed boundary conditions for the gauged model.

\subsection{A construction from CSS stabilizer Hamiltonians} \label{CSS}

In light of the above discussion it is clear that from a complex corresponding to a topological CSS stabilizer code
\begin{align}
\begin{CD}
\hmodule @>\ham>> \pmodule @>\conj>> \emodule
\end{CD}
\end{align}
where $\sigma=\sigma_X \oplus \sigma_Z$, one can read off a gauging duality. 
This duality is specified in our language by the maps $\zfields=\sigma_X^\dagger$ and $\kernn = \sigma_Z$. Hence the ungauged Hamiltonian is generated by $\xfields\oplus\sigma_X^\dagger$ (with $\text{im}(\xfields)=\text{ker}(\sigma_X)$) with local symmetric $X$ field perturbations $(\openone_Q , \, \vline\, \mathbf{0} )$. 

From this analysis we see that if $\text{ker}(\sigma_X)$ is locally trivial, in the sense that it contains no local elements,  then the ungauged Hamiltonian possesses only global symmetries and the stabilizers are all $Z$ fields (see the examples in Section~\ref{examples}). 
This point highlights a difference between the cubic code and generalized toric codes, while both the respective ungauged variants may have a growing number of global $X$ symmetries (one for each redundant $X$ stabilizer) the former has no local symmetries whereas the latter has an extensive number. This is relevant to the distinct behaviours of their ground state degeneracies. We speculate that it is indicative of spatially extended vs. isolated pointlike excitations

\subsection{Cluster state construction \& gauging}

We now go slightly beyond CSS stabilizers and consider cluster state models built on bipartite graphs specified by the map $\zfields$ from the gauging procedure for some CSS Hamiltonian.  By construction this cluster state inherits the $X$ symmetry of the input ungauged model (corresponding to $\text{ker}(\zfields^\dagger)$) on one sublattice and an $X$ symmetry on the other sublattice in the position of each $Z$ symmetry of the input gauged model. 
This cluster model is clearly short-range entangled (SRE) since it can be mapped to a trivial decoupled model via a local circuit of $CZ$s. However this disentangling does not respect the symmetries. Hence these cluster states are candidates for higher form or fractal symmetry-protected topological (SPT)~\cite{1Done,chen2013symmetry,chen2012symmetry,schuch2011classifying,pollmann2010entanglement} states. 

Several different approaches could be taken when gauging these cluster models. 
We take advantage of the natural bipartite structure of the system and treat the two disjoint sublattices separately. Since the terms appearing on a single sublattice are either single $X$s or products of $Z$s, generated by $\zfields$ or $\zfields^\dagger$ respectively, one can instantly read off the effect of gauging one sublattice. 
Specifically it results in a doubling of the qubits on the remaining sublattice with each $X,Z$ field on that lattice now accompanied by a $Z,X$ term, respectively, on the new partner qubit. 

$Z$ fields generated by either $\kernn$ or $\xfields$, depending on the sublattice gauged, are also added to the new qubits. These intermediate models can possess topological order since they are equivalent to either the input gauged or ungauged model under a local circuit of $CZ$s. However these gates do not respect the symmetry on the remaining sublattice, which is indicative of the possibility of symmetry-enriched topological (SET)~\cite{bombin2010topological,mesaros2013classification,hung2013quantized,barkeshli2014symmetry,tarantino2015symmetry,kirillov2004g,turaev2000homotopy} order. 

Once both sublattices have been gauged one can easily see that the model is mapped to itself up to local swaps and Hadamards.

\section{Fractal Symmetries}
\label{examples}

In this section we present several examples which consist of pairs of models that are dual under gauging and support interesting symmetries.

\subsection{2D toric code - Ising model}
\label{toric code}
The toric code is a CSS code that is known to be equivalent to a $\mathbb{Z}_2$ gauge theory (see Section 3 of Ref.\cite{kitaev2003fault} and also Ref.\cite{kitaev2010topological}).
In the polynomial language the toric code is generated by ${\ham_X=(x+xy,y+xy)},\ \ham_Z=(1+x,1+y)$, graphically
\begin{align}
 \xymatrix@!0{
 IX \ar@{-}[r]  &XX \ar@{-}[d]\\
II \ar@{-}[u]   &XI \ar@{-}[l] 
 }
 \quad \quad
 \xymatrix@!0{
 IZ \ar@{-}[r]  &  II \ar@{-}[d] \\
 ZZ \ar@{-}[u]   &ZI \ar@{-}[l]
 } .
\end{align}
The map $\zfields=\sigma_X^\dagger$ determined by the $X$ stabilizers corresponds to the 2D Ising model, generated by the terms $(0,1+y),\ (0,1+x)$
\begin{align}\label{sixtn}
 \xymatrix@!0{
 Z \ar@{-}[r]  &  I \ar@{-}[d] \\
 Z \ar@{-}[u]   & I \ar@{-}[l]
 } 
 \quad \quad
 \xymatrix@!0{
 I \ar@{-}[r]  & I \ar@{-}[d]\\
Z \ar@{-}[u]   & Z \ar@{-}[l] 
 }
\end{align}
with $X$ perturbations generated by $(1,0)$.

Note that when gauging the Ising model one encounters a situation in which $\text{ker}(\zfields)$ is nontrivial and is generated by $\ham_Z$ as expected. To see this explicitly consider a product of terms from Eq.\eqref{sixtn} around a plaquette that yields the identity.

\subsection{3D cubic code - fractal symmetry Ising model}
\label{cubic code}

The cubic code is a 3D CSS code generated by 
$$
\ham_X=
\begin{pmatrix}
x+y+z+xyz \\
1+y+xy+yz
\end{pmatrix} \,
\ham_Z=
\begin{pmatrix}
x+y+xy+xyz \\
1+xy+yz+xz
\end{pmatrix}
$$
or graphically
\begin{align}
\drawgenerator{XI}{II}{IX}{XI}{IX}{XX}{XI}{IX}
\quad
\drawgenerator{ZI}{ZZ}{IZ}{ZI}{IZ}{II}{ZI}{IZ}
\end{align}
the gauging map $\zfields$ specifies the ungauged cubic code, a type of `Ising' model, with $Z$ stabilizers generated by 
$${(0,1+xy+xz+yz)},\ {(0,x+z+xz+xyz)}$$
\begin{align}
\drawgenerator{I}{Z}{Z}{I}{Z}{I}{I}{Z}
\quad
\drawgenerator{Z}{Z}{I}{Z}{I}{I}{Z}{I}
\end{align}
and single qubit $X$ perturbations.
This model has a fractal $X$ symmetry for each product of stabilizers equal to the identity in the cubic code.

\subsection{Self dual cluster models}
\label{fractal cluster}

The first cluster model is derived from the 2D toric code and has stabilizers
\begin{align}
& \xymatrix@!0{
 \red{IZ} \ar@{.}[rr] &  &  \red{II} \ar@{.}[dd] \\
 &  \blue{(X)} \ar@{-}[d] \ar@{-}[u] \ar@{-}[l] \ar@{-}[r] & \\
 \red{ZZ} \ar@{.}[uu] &  &\red{ZI} \ar@{.}[ll]
 } 
 \\
& \hspace{.15cm}  \xymatrix@!0{
 \blue{Z} \ar@{-}[rr] & &  \blue{I} \ar@{-}[dd] \\
 &  \red{(XI)} \ar@{.}[d] \ar@{.}[u] \ar@{.}[l] \ar@{.}[r] & \\
 \blue{Z} \ar@{-}[uu] &  &\blue{I} \ar@{-}[ll]
 } 
 \quad \quad
   \xymatrix@!0{
 \blue{I} \ar@{-}[rr] & &  \blue{I} \ar@{-}[dd] \\
 &  \red{(IX)} \ar@{.}[d] \ar@{.}[u] \ar@{.}[l] \ar@{.}[r]  & \\
 \blue{Z} \ar@{-}[uu] &  &\blue{Z} \ar@{-}[ll]
 } 
 \end{align}
where the matter sublattice has a single qubit per site (blue) and the gauge sublattice has two (red). 
This model has a 1D (1-form) $X$ symmetry on the red sublattice and a global (0-form) $X$ symmetry on the blue sublattice.

The previous example fits into a broad class of cluster states in arbitrary dimension $\sdim$ with qubits on $(k-1)$- and $(k)$- cells. These cluster states are constructed on the bipartite adjacency graph of these cells and possess $(d-k)$- and $(k-1)$- form $X$ symmetry. 
\\

The second cluster state model comes from the cubic code and has stabilizers
\begin{align}
&\drawgeneratorclustertwoa{\red{ZI}}{\red{ZZ}}{\red{IZ}}{\red{ZI}}{\red{IZ}}{\red{II}}{\red{ZI}}{\red{IZ}}{ \blue{X}}
\\
&\drawgeneratorclustertwob{\blue{I}}{\blue{I}}{\blue{Z}}{\blue{I}}{\blue{Z}}{\blue{Z}}{\blue{I}}{\blue{Z}}{\red{XI}} \quad
\drawgeneratorclustertwob{\blue{Z}}{\blue{I}}{\blue{I}}{\blue{Z}}{\blue{I}}{\blue{Z}}{\blue{Z}}{\blue{I}}{\red{IX}}
\end{align}
using the same sublattice conventions as above. One can see by inspecting the pictures that all translations of these terms commute.

This model supports fractal $X$ symmetries on each sublattice. It inherits an $X$ symmetry on the red sublattice for each $X$ symmetry of the cubic code and an $X$ symmetry on the blue sublattice for each $X$ symmetry of the ungauged cubic code.

\section{ Conclusions}

In this paper we have defined a gauging procedure for general submanifold symmetries, including those of fractal type, within the framework of Pauli Hamiltonians. 
We demonstrated relations between the pre- and post- gauging models reminiscent of those obtained via the conventional gauging procedure. 
Using the tools developed in this process we constructed short-range entangled Ising and cluster models with fractal symmetries and examined their transformation under gauging. 

This gauging procedure constitutes a small step towards adapting the standard tools from the condensed matter toolbox for application to more exotic 3D topological orders, including cases where the common sense assumptions leading to a TQFT description are not satisfied~\cite{haah2014bifurcation}. 
We are optimistic that this path leads to a cache of strange and exotic phases of matter beyond (conventional) TQFTs, citing Haah's cubic code~\cite{haah2011local} as a demonstrative example. 

Our approach opens the door to more general constructions and a possible relation between SRE fractal-symmetric and exotic topological phases similar to the well known connection between SPT phases and Dijkgraaf-Witten theories~\cite{DijkgraafWitten,williamson2014matrix}. 
In particular the gauging procedure applied to a subgroup of the global symmetry allows one to construct and study fractal SET phases. The fractal symmetries in this context may play a role in understanding the most general transversal gates in topological codes~\cite{beigi2011quantum,yoshida2015gapped} and a connection of these phases of matter to quantum computation~\cite{else2012symmetry,williamson2015symmetry}.
\\
\\
\\
\emph{Acknowledgments -} The author acknowledges Frank Verstraete, Norbert Schuch and especially Jeongwan Haah for helpful discussions (and for the use of his figures), Jacob Bridgeman and Sam Roberts for comments, trivia with the CB and is grateful to the Berkeley Simons Institute and the organizers of the Quantum Hamiltonian Complexity program where this work originated and also to the quantum theory group at the University of Sydney where it was written up.

\bibliography{CC}

\end{document}